\documentclass[runningheads]{llncs}

\usepackage{appendix}
\usepackage{color,soul}
\usepackage{graphicx}

\newcommand{\prob}[1]{\vspace{3mm}\noindent\fbox{\parbox{\textwidth}{#1}}\vspace{3mm}}

\newtheorem{reduction}{\indent Reduction Rule}

\begin{document}
\title{A Discharging Method: Improved Kernels for Edge Triangle Packing and Covering}
\titlerunning{Improved Kernels for Edge Triangle Packing and
Covering}
%


\author{Zimo Sheng \and
Mingyu Xiao\orcidID{0000-0002-1012-2373}\thanks{Corresponding author}}
%
%
\authorrunning{Zimo Sheng\inst{1} \and
Mingyu Xiao\inst{2}}
\institute{School of Computer Science and Engineering, University of Electronic Science and Technology of China, Chengdu, China\\
\email{shengzimo2016@gmail.com}, \email{myxiao@uestc.edu.cn}}
\maketitle              
%


    \begin{abstract}
        \textsc{Edge Triangle Packing} and \textsc{Edge Triangle Covering} are dual problems extensively studied in the field of parameterized complexity.
        Given a graph $G$ and an integer $k$, \textsc{Edge Triangle Packing} seeks to determine whether there exists a set of at least $k$ edge-disjoint triangles in $G$,
        while \textsc{Edge Triangle Covering} aims to find out whether there exists a set of at most $k$ edges that intersects all triangles in $G$.
      Previous research has shown that \textsc{Edge Triangle Packing} has a kernel of $(3+\epsilon)k$ vertices, while \textsc{Edge Triangle Covering} has a kernel of $6k$ vertices.
        In this paper, we show that the two problems allow kernels of $3k$ vertices, improving all previous results. A significant contribution of our work is the utilization of a novel discharging method for analyzing kernel size, which exhibits potential for analyzing other kernel algorithms.
    \end{abstract}

\section{Introduction}
Preprocessing is a fundamental and commonly used step in various algorithms.
However, most preprocessing has no theoretical guarantee on the quality.
Kernelization, originating from the field of parameterized algorithms~\cite{DBLP:books/sp/CyganFKLMPPS15},
now has been found to be an interesting way to analyze the quality of preprocessing.
Consequently, kernelization has received extensive attention in both theoretical and practical studies.

Given an instance $(I,k)$  of a problem, a kernelization (or a kernel algorithm)
runs in polynomial time and returns an equivalent instance $(I',k')$ of the same problem such that $(I,k)$
 is a yes-instance if and only if $(I',k')$
 is a yes-instance, where $k'\leq k$
 and $|I'|\leq g(k)$  for some computable function $g$ only of $k$.
 The new instance $(I',k')$  is called a \emph{kernel} and $g(k)$
 is the size of the kernel. If g(·) is a polynomial or linear function, we classify the problem as having a polynomial or linear kernel, respectively.

\textsc{Edge Triangle Packing} (ETP), to check the existence of $k$ edge-disjoint triangles in a given graph $G$ is NP-hard even on planar graphs with maximum degree 5~\cite{DBLP:journals/siamcomp/Holyer81}. The optimization version of this problem is APX-hard on general graphs~\cite{DBLP:journals/ipl/Kann94}.
A general result of~\cite{DBLP:journals/siamdm/HurkensS89} leads to a polynomial-time $(3/2+\epsilon)$
approximation algorithm for any constant $\epsilon>0$. When the graphs are restricted to planar graphs, the result can be improved.
A polynomial-time approximation scheme for the vertex-disjoint triangle packing problem on planar graphs was given by~\cite{DBLP:journals/jacm/Baker94}, which can be extended to ETP on planar graphs.
In terms of parameterized complexity, a $4k$-vertex kernel and an $O^*(2^{\frac{9k}{2}\log k+\frac{9k}{2}})$-time parameterized algorithm for ETP were developed in~\cite{DBLP:conf/iwpec/MathiesonPS04}. Later, the size of the kernel was improved to $3.5k$~\cite{DBLP:journals/ipl/Yang14}. The current best-known result is $(3+\epsilon)k$~\cite{DBLP:journals/ipl/LinX19}, where $\epsilon>0$ can be any positive constant. On tournaments, there is also a kernel of $3.5k$ vertices~\cite{DBLP:journals/chinaf/YuanFW23}.

Another problem considered in this paper is \textsc{Edge Triangle Covering} (ETC). ETC is the dual problem of ETP, which is to check whether we can delete at most $k$ edges from a given graph such that the remaining graph has no triangle. ETC is also NP-hard even on planar graphs with maximum degree 7~\cite{DBLP:journals/endm/BrugmannKM09}.
In terms of kernelization, a $6k$-vertex kernel for ETC was developed~\cite{DBLP:journals/endm/BrugmannKM09}. On planar graphs, the result was further improved to $\frac{11k}{3}$~\cite{DBLP:journals/endm/BrugmannKM09}.

In this paper, we will deeply study the structural properties of \textsc{Edge Triangle Packing} and \textsc{Edge Triangle Covering} and give some new reduction rules by using a variant of crown decomposition. After that, we will introduce a new technology called the discharging method to analyze the size of problem kernels. Utilizing the new discharging method, we obtain improved kernel sizes of $3k$ vertices for both ETP and ETC. Notably, our results even surpass the previously best-known kernel size for ETC on planar graphs~\cite{DBLP:journals/endm/BrugmannKM09}.

\vspace{-4mm}
\section{Preliminaries}
Let $G=(V,E)$ denote a simple and undirected graph with $n=|V|$ vertices and $m=|E|$ edges.
A vertex is a \emph{neighbor} of another vertex if there is an edge between them. The set of neighbors of a vertex $v$ is denoted by $N(v)$,
and the degree of $v$ is defined as $d(v) = |N(v)|$.
For a vertex subset $V'\subseteq V$, we let $N(V')=\cup_{v\in V'} N(v)\setminus V'$ and $N[V']=N(V')\cup V'$.
The subgraph induced by a vertex subset $V'\subseteq V$ is denoted by $G[V']$ and the subgraph spanned by an edge set $E'\subseteq E$ is denoted by $G[E']$.
The vertex set and edge set of a graph $H$ are denoted by $V(H)$ and $E(H)$, respectively.

A complete graph on $3$ vertices is called a \emph{triangle}.
We will use $vuw$ to denote the triangle formed by vertices $v$, $u$, and $w$. If there is a triangle $vuw$ in $G$, we say that vertex $v$ \emph{spans} edge $uw$.
An \emph{edge triangle packing} in a graph is a set of triangles such that every two triangles in it have no common edge.
The \textsc{Edge Triangle Packing} problem (ETP) is defined below.

\vspace{-2mm}
\prob{
\textsc{Edge Triangle Packing} (ETP)      ~~~\textbf{Parameter:} $k$\\
\textbf{Input}: An undirected  graph $G=(V,E)$, and an integer $k$.\\
\textbf{Question}: Does there exist an edge triangle packing of size at least $k$ in $G$?}

An edge \emph{covers} a triangle if it is contained in the triangle. An \emph{edge triangle covering} in a graph is a set of edges $S$ such that there is no triangle after deleting $S$ from $G$.
The \textsc{Edge Triangle Covering} problem (ETC) is  defined below.

\vspace{-2mm}
\prob{
\textsc{Edge Triangle Covering} (ETC)      ~~~\textbf{Parameter:} $k$\\
\textbf{Input}: An undirected  graph $G=(V,E)$, and an integer $k$.\\
\textbf{Question}: Does there exist an edge triangle covering of size at most $k$ in $G$?}

\section{Fat-Head Crown Decomposition}
One important technique in this paper is based on a variant of the crown decomposition.
Crown decomposition is a powerful technique for the famous \textsc{Vertex Cover} problem and it has been extended to solve several related problems~\cite{DBLP:conf/iwpec/DehneFRS04,DBLP:journals/tcs/XiaoK20,DBLP:conf/tamc/XiaoK17,DBLP:conf/mfcs/CervenyCS22}.
Specifically, we employ a specific variant called the fat-head crown decomposition to tackle (ETP)~\cite{DBLP:journals/ipl/LinX19}. This variant of the crown decomposition will also be applied in our algorithms for both ETP and ETC. To provide a comprehensive understanding, let us begin by introducing the definition of the fat-head crown decomposition.

A \emph{fat-head crown decomposition} of a graph $G=(V,E)$  is a triple $(C,H,X)$ such that $C$ and $X$ form a partition of $V$ and $H\subseteq E$ is a subset of edges satisfying the following properties:\\
1. $C$ is an independent set.\\
2. $H$ is the set of edges spanned by at least one vertex in $C$.\\
3. No vertex in $C$ is adjacent to a vertex in $X\setminus V(H)$. \\
4. There is an edge-disjoint triangle packing $P$ of size $|P|=|H|$ such that each triangle in $P$  contains exactly one vertex in $C$ and exactly one edge in $H$. The packing $P$ is also called the \emph{witness packing} of the fat-head crown decomposition.

%
%
%

\begin{figure}[t]
    \centering
    \includegraphics[scale=0.22]{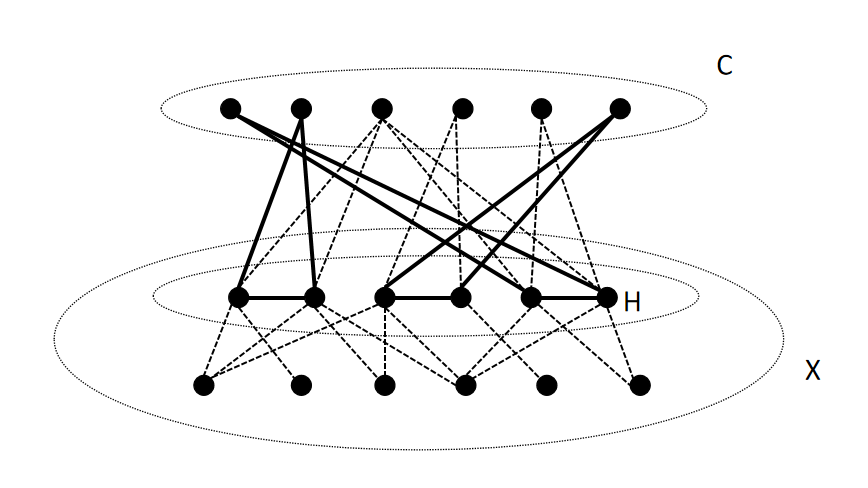}
    \caption{An illustration for the fat-head crown decomposition}
    \label{FCD}
    \vspace{-2mm}
\end{figure}

An illustration of the fat-head crown decomposition is shown in Figure~\ref{FCD}. To determine the existence of fat-head crown decompositions in a given graph structure, we present three lemmas.

\begin{lemma}\label{lemma_cr}\emph{(Lemma 2 in~\cite{DBLP:journals/ipl/LinX19})}
 Let $G=(V,E)$ be a graph such that each edge and each vertex is contained in at least one triangle. Given a  non-empty independent set $I\subseteq V$ such that $|I|>|S(I)|$, where $S(I)$  is the set of edges spanned by at least one vertex in $I$. A fat-head crown decomposition $(C,H,X)$ of $G$ with $C\subseteq I$ and $H\subseteq S(I)$ together with a witness packing $P$ of size $|P|=|H|>0$ can be found in polynomial time.
\end{lemma}

\begin{lemma}\label{FCD_exist}
Given a graph $G=(V,E)$, a vertex set $A\subseteq V$, and an edge set $B\subseteq E$, where $A\cap V(B)=\emptyset$.
There is a polynomial-time algorithm that checks whether there is a fat-head crown decomposition $(C,H,X)$ such that $\emptyset \neq C \subseteq A$ and $H\subseteq B$ and outputs one if yes.
\end{lemma}
\begin{proof}
We begin by demonstrating that if there exists an edge $uv \in E$ between two vertices $u$ and $v \in A$, neither $u$ nor $v$ can belong to $C$ in the fat-head crown decomposition. Suppose, without loss of generality, that $u$ is in $C$. Since $C$ is an independent set, $v$ cannot be in $C$. Moreover, $v$ cannot be in $X \setminus V(H)$ since no vertex in $C$ is adjacent to a vertex in $X \setminus V(H)$. Additionally, $v$ cannot be in $V(H)$ since $A \cap V(B) = \emptyset$. Thus, this scenario is impossible. Therefore, neither $u$ nor $v$ can be in $C$ for any fat-head crown decomposition $(C, H, X)$ with $C \subseteq A$ and $H \subseteq B$. Consequently, we can exclude vertices in $A$ that are adjacent to other vertices in $A$ when computing the decomposition.

To facilitate the computation, we construct an auxiliary bipartite graph $G' = (A', B', E')$, where each vertex in $A'$ corresponds to a vertex in $A$ that is not adjacent to any other vertex in $A$, each vertex in $B'$ corresponds to an edge in $B$, and an edge exists between $w \in A'$ and $e \in B'$ if and only if $w$ spans $e$ in $G$. Based on the construction of $G'$, we deduce that a fat-head crown decomposition $(C, H, X)$ exists in $G$ if and only if there exist $C' \subseteq A'$ and $H' \subseteq B'$ in $G'$ such that there is a matching of size $|H'|$ in the induced bipartite graph $G'[C' \cup H']$.

The existence of the matching can be determined in $O(mn^{1.5})$ time using the expansion lemma~\cite{DBLP:journals/jacm/BodlaenderFLPST16}. Thus, the algorithm can verify the existence of a fat-head crown decomposition and output one if it exists, all within polynomial time complexity.
\hfill \qed\par
\end{proof}


\begin{lemma}\label{FCD_correct}
    If there is a fat-head crown decomposition $(C,H,X)$ in $G$, then $G$ has an edge-disjoint triangle packing (resp., edge triangle covering) of size $k$ if and only if the graph $G'$ has
     an edge-disjoint triangle packing (resp., edge triangle covering) of size $k-|H|$, where $G'$ is the graph obtained from $G$ by deleting vertex set $C$ and deleting edge set $H$.
\end{lemma}
\begin{proof}
For a fat-head crown decomposition $(C,H,X)$, there is witness packing $P$ of size $|P|=|H|$.

We first show that there is an edge-disjoint triangle packing $D$ of size $k$ in $G$ if and only if there is an edge-disjoint triangle packing $D'$ of size $k-|H|$ in $G'$. We only prove the sufficient condition and the other direction is obvious since $D'\cup P$ is a valid edge-disjoint triangle packing in $G$.
Let $D^*\subseteq D$ be the set of triangles containing either a vertex in $C$ or an edge in $H$. For the former case, triangle $T$ still contains one edge in $H$ since any triangle containing
a vertex in $C$ must contain one edge in $H$. Thus, each triangle in $D^*$ contains at least one edge in $H$, which means
that $|D^*|\leq |H|$. Therefore, $D-D^*$ is an edge-disjoint
triangle packing of size at least $k-|H|$ in graph $G'$.

Next, we prove that there is an edge triangle covering $M$ of size  $k$ in $G$ if and only if there is an edge triangle covering $M'$  of size $k-|H|$ in $G'$.
For the necessary condition. We prove that $M'\cup H$ is a valid edge triangle covering in $G$.
Consider an arbitrary  triangle $T'$ in $G$.
If $V(T')\cap V(C)=\emptyset$ and $E(T')\cap E(H)=\emptyset$, then $T'$ is covered by $M'$ since $M'$ is an edge triangle covering for $G'$.
If $V(T')\cap C\neq \emptyset$, then $T'$ is covered by $H$ since vertices in $C$ can only span edges in $H$ by the definition of the fat-head crown decomposition.
If $E(T')\cap E(H)\neq \emptyset$, then $T'$ is covered by $H$. For any case, triangle $T'$ is covered by at least one edge in $M'\cup H$.
For the sufficient condition. Let $M^*= M\cap E(P)$. We know that $|M^*|\geq |P|=|H|$ since $P$ is an edge-disjoint triangle packing  of size $|H|$. Thus, $M-M^*$ is an edge triangle covering of size at most $k-|H|$ in $G'$.
\hfill \qed\par
\end{proof}
\section{The Algorithms}
In this section, we present our kernelization algorithms for the \textsc{Edge Triangle Packing} (ETP) and \textsc{Edge Triangle Covering} (ETC) problems.
Our algorithms involve a set of reduction rules that are applied iteratively until no further reduction is possible. Each reduction rule is applied under the assumption that all previous reduction rules have already been applied and cannot be further applied to the current instance.
A reduction rule is \emph{correct} if the original instance $(G,k)$ is a yes-instance if and only if the resulting instance $(G',k')$ after applying the reduction
rule is a yes-instance.

We have one algorithm for ETP and ETC, respectively. The two algorithms are similar. We will mainly describe the algorithm for ETP and introduce the difference for ETC.
In total, we have nine reduction rules. The first four rules are simple rules to handle some special structures, while the remaining five rules are based on a triangle packing. Especially, the last rule will use the fat-head crown decomposition. We will show that the algorithms run in polynomial time.

\subsection{Simple Rules}

\begin{reduction}\label{rule0} For ETP,  if $k\leq0$, then return `yes' to indicate that the instance is a yes-instance;
if $k> 0$ and the graph is empty, then return `no' to indicate that the instance is a no-instance.\\
 For ETC,
if $k\geq0$ and the graph is empty, then return `yes' to indicate that the instance is a yes-instance;
if $k<0$, then return `no' to indicate that the instance is a no-instance.
\end{reduction}

\begin{reduction}\label{rule1} If there is a vertex or an edge not appearing in any triangle, then delete it from the graph.
\end{reduction}


\begin{reduction}\label{4v} If there are 4 vertices $u,v,w,x\in V$ inducing a complete graph (i.e., there are 6 edges $uv,uw,ux,vw,vx,wx\in E$) such that
none of the 6 edges is in a triangle except $uwv,uvx,uwx$, and $vwx$, then \\
- For ETP, delete the 6 edges $uv,uw,ux,vw,vx$ and $wx$ and let $k=k-1$;\\
- For ETC, delete the 6 edges $uv,uw,ux,vw,vx$ and $wx$ and let $k=k-2$.
\end{reduction}

The correctness of Reduction Rule~\ref{4v} is based on the following observation. For ETP, any edge triangle packing can have at most one triangle containing some edge from these 6 edges and we can simply take one triangle from this local structure. For ETC, any edge triangle covering must contain at least two edges from these 6 edges and after deleting $vu$ and $wx$, none of $uw,ux,vw$, and $vx$ is contained in a triangle anymore.


\begin{reduction} \label{split}
    If there is a vertex $v\in V$ such that all edges incident to $v$ can be partitioned into two parts $E_1$ and $E_2$ such no triangle in $G$ contains an edge in $E_1$ and an edge in $E_2$, then split $v$ into two vertices $v'$ and $v''$ such that all edges in $E_1$ are incident on $v'$ and all edges in $E_2$ are incident on $v''$.
\end{reduction}
An illustration of Reduction Rule~\ref{split} is shown in Figure~\ref{fsplit}. This reduction rule will increase the number of vertices in the graph. However, this operation will simplify the graph structure and our analysis.
\begin{figure}[t]
    \centering
   \includegraphics[scale=0.30]{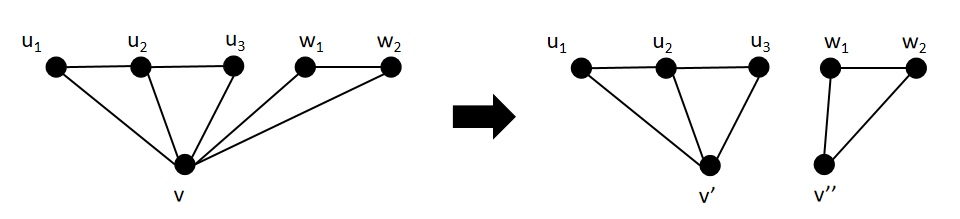}
   \caption{An Illustration for Reduction Rule \ref{split}}
   \label{fsplit}
\end{figure}
\begin{lemma}\label{lem_split}
    Reduction Rule~\ref{split} is correct and can be executed in polynomial time.
\end{lemma}
\begin{proof} First, we consider the correctness.
    Let $G'=(V',E')$ be the graph after applying Reduction Rule~\ref{split} on a vertex $v$.
    We can establish a one-to-one mapping between the edges in $E$ and the edges in $E'$ by considering the vertices $v'$ and $v'' \in V'$ as $v \in V$. Three edges in $E$ form a triangle in $G$ if and only if the three corresponding edges in $E'$ form a triangle in $G'$ since there is no triangle in $G$ contains an edge in $E_1$ and an edge in $E_2$.
          Thus, an edge triangle packing of size $k$ (resp., an edge triangle covering of size $k$) in $G$ is also an edge triangle packing of size $k$ (resp., an edge triangle covering of size $k$) in $G'$.
          This implies that Reduction Rule~\ref{split} is correct for both ETP and ETC.

    We give a simple greedy algorithm to find the edge sets $E_1$ and $E_2$ for a given vertex $v$.
    Initially, let $E_1$ contain an arbitrary edge $e$ incident on $v$. We iteratively perform the following steps until no further updates occur: if there is a triangle containing an edge in $E_1$ and an edge $e'$ incident on $v$ but not in $E_1$, then add edge $e'$ to $E_1$.
    It is easy to see that all edges in $E_1$ must be in the same part to satisfy the requirement. If $E_1\neq E$, then we can split $E$ to two parts $E_1$ and $E_2=E\setminus E_1$.
    Otherwise, the edges incident on $v$ cannot be split.
\hfill \qed\par\end{proof}

\subsection{Adjustments Based on a Triangle Packing}
After applying the first four rules, our algorithms will find a maximal edge-disjoint triangle packing $S$ by using an arbitrary greedy method. This can be done easily in polynomial time.
The following rules are based on the packing $S$.
From now on, we let $F = V \setminus V (S)$ denote the set of vertices not appearing in $S$ and $R= E\setminus E(S)$ denote the set of edges not appearing in $S$.
We begin with the following trivial rule.

\begin{reduction}\label{packing}
 If $|S|> k$,  for ETP, return `yes' to indicate that the instance is a yes-instance; and for ETC, return `no' to indicate that the instance is a no-instance.
\end{reduction}

The following three rules just update the packing $S$ by replacing some triangles in it and do not change the graph.
Illustrations of the three rules are shown in Figure~\ref{567}.

\begin{figure}[h]
    \centering
    \includegraphics[scale=0.22]{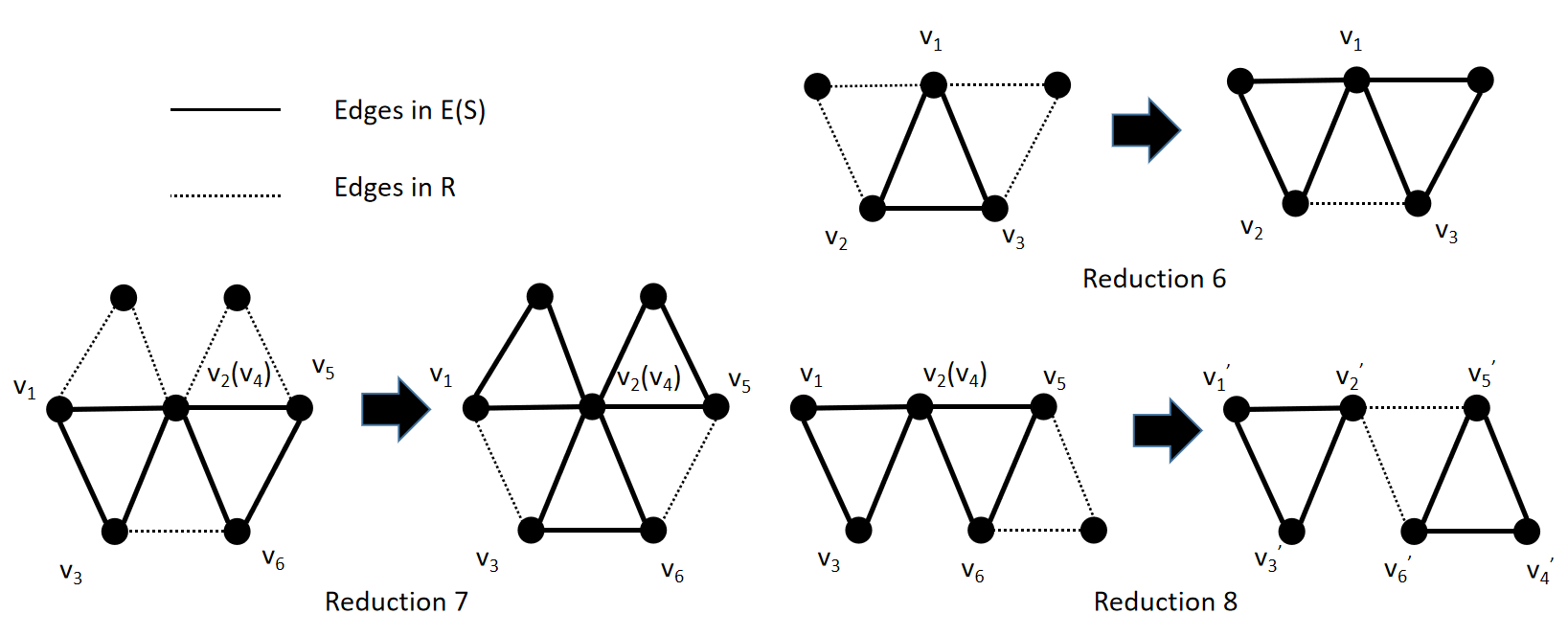}
    \caption{Illustrations of Reduction Rules \ref{increasing}-\ref{increasing3}}
    \label{567}
\end{figure}

\begin{reduction}\label{increasing}
 If there is a triangle $v_1v_2v_3 \in S$ such that there are at least two edge-disjoint triangles in the spanned graph $G[R\cup \{v_1v_2,v_1v_3,v_2v_3\}]$, then replace $v_1v_2v_3$ with these triangles in $S$ to increase the size of $S$ by at least 1.
\end{reduction}

\begin{reduction}\label{increasing2}
 If there are two edge-disjoint triangles $v_1v_2v_3$ and $v_4v_5v_6\in S$  such that there are at least three edge-disjoint triangles in the spanned graph $G[R\cup \{v_1v_2,v_1v_3,v_2v_3,v_4v_5,v_4v_6,v_5v_6\}]$, then replace $v_1v_2v_3$ and $v_4v_5v_6$ with these triangles in $S$ to increase the size of $S$ by at least 1.
\end{reduction}

\begin{reduction}\label{increasing3}
 If there are two edge-disjoint triangles $v_1v_2v_3$ and $v_4v_5v_6\in S$  such that there are  two edge-disjoint triangles $v'_1v'_2v'_3$ and $v'_4v'_5v'_6$ in the induced graph $G[F\cup \{v_1,v_2,v_3,v_4,v_5,v_6\}]$ such that $|\{v'_1, v'_2, v'_3\} \cup \{v'_4, v'_5, v'_6\}| > |\{v_1, v_2, v_3\}\cup \{v_4, v_5, v_6\}|$ , then replace triangles $v_1v_2v_3$ and $v_4v_5v_6$ with  triangles $v'_1v'_2v'_3$ and $v'_4v'_5v'_6$ in $S$ to increase the number of vertices appearing in $S$ by at least 1.
\end{reduction}
 Note that an application of Reduction Rules \ref{increasing}-\ref{increasing3} will not change the structure of the graph. Thus, the first four reduction rules
 will not be applied after executing Reduction Rules~\ref{increasing}-\ref{increasing3}.

\subsection{A Reduction Based on Fat-head Crown Decomposition}

After Reduction Rule~\ref{increasing3}, we obtain the current triangle packing $S$.
An edge in $E(S)$ is called a \emph{labeled edge} if it is spanned by at least one vertex in $F$.
We let $L$ denote the set of labeled edges.

We can find a fat-head crown decomposition $(C,H,X)$ with $C \subseteq V\setminus V(L)$ and $H\subseteq L$ in polynomial time if it exists by Lemma \ref{FCD_exist}.
Moreover, we will apply the following reduction rule to reduce the graph, the correctness of which is based on Lemma \ref{FCD_correct}.
\begin{reduction}\label{crown}
Use the algorithm in Lemma \ref{FCD_exist} to check whether there is a fat-head crown decomposition $(C, H, X)$ such that $\emptyset \neq C \subseteq V\setminus V(L)$ and $H\subseteq L$. If yes, then delete vertex set $C$ and edge set $H$, and let $k=k-|H|$.
\end{reduction}


An instance is called \emph{reduced} if none of the nine reduction rules can be applied to it. The corresponding graph is also called a reduced graph.

\begin{lemma}\label{time}
    For any input instance, the kernelization algorithms run in polynomial time to output a reduced instance.
\end{lemma}
\begin{proof}
Lemma \ref{FCD_exist} demonstrates that one execution of Reduction Rule \ref{crown} can be performed within polynomial time, while Lemma \ref{lem_split} shows that one execution of Reduction Rule \ref{split} can also be accomplished in polynomial time. Additionally, each of the remaining reduction rules can be easily executed within polynomial time.
Next, we only need to show that each reduction rule can be executed for at most polynomial times to complete our proof. Reduction Rules \ref{rule0} and \ref{packing} can be applied for at most 1 time.
Reduction Rule \ref{rule1} can be applied for at most $m$ times since it will delete at least one edge each time.
Reduction Rule \ref{split} can be applied for at most $2m$ times since it will increase one vertex each time and there are at most $2m$ vertices for the worst case.
Reduction Rules \ref{4v} and \ref{crown} can be applied for at most $k$ times since each of them will decrease $k$ by at least 1.
We consider Reduction Rules \ref{increasing}-\ref{increasing3}. For a fixed graph $G=(V,E)$, Reduction Rules \ref{increasing} and \ref{increasing2} can be continuously executed for at most $k$ times
 and  Reduction Rule \ref{increasing3} can be continuously executed for at most $3k$ times.
 Since Reduction Rule \ref{crown} can be applied for at most $k$ times to change the graph, we know that Reduction Rules \ref{increasing}-\ref{increasing3} can be applied for at most $O(k^2)$ times.
  Thus, the algorithms run in polynomial times.
\hfill \qed\par
\end{proof}

\section{Analysis Based on Discharging}

Next, we use a discharging method to analyze the size of a reduced instance.
Note that there is no significant difference between ETC and ETP in the analysis.
We partition the graph into two parts: one part is the edge-disjoint triangle packing $S$ after applying all the reductions; the other part is the set $F$ of vertices not appearing in $S$.
Before proceeding with the analysis, we will establish some properties that will be utilized.


\begin{lemma}\label{Loe}
    Consider a reduced graph $G=(V,E)$ with triangle packing $S$. For any triangle $uvw\in S$, at most one of $\{uv,vw,uw\}$ is a labeled edge.
\end{lemma}
\begin{proof}
    Assume to the contrary that there are two edges, say $uv$ and $vw$ are spanned by vertices in $F$. We show some contradiction.

    If edges $uv$ and $vw$ are spanned by two different vertices $x,x'\in F$ respectively, then Reduction Rule \ref{increasing} could be applied (Case 1 in Figure \ref{f3}).
   Therefore, edges $uv$ and $vw$ are spanned by the same vertex $x\in F$.
   Since Reduction Rule \ref{4v} is not applied on the four vertices $\{u,v,w,x\}$, we know that at least one edge in $\{uw,uw,ux,vw,vx,wx\}$ is contained in a triangle other than $uwv,uvx,uwx,$ and $vwx$.
   Due to symmetry, we only need to consider two edges $vw$ and $xw$.

   Assume that edge $vw$ is contained in a triangle $vwy$, where $y\not\in \{u, x\}$. If none of $\{yv,yw\}$ appears in $E(S)$, then Reduction Rule \ref{increasing} could be applied to replace $uvw$ with two triangles $xvu$ and $yvw$ in $S$. If at least one edge in  $\{yv, yw\}$ is contained in $E(S)$, without loss of generality, assume $yw\in E(S)$ and there is a triangle $ywz\in S$.
   For this case, Reduction Rule \ref{increasing3} could be applied to replace $vuw$ and $ywz$ with $xvu$ and $ywz$ (Case 2 in Figure \ref{f3}).

Assume that edge $xw$ is contained in a triangle $xwy$, where $y\not\in \{u, v\}$. By the maximality of $S$, we know that at least one of $\{xy, yw\}$ must appear in $E(S)$.
However, edge $xy$ can not appear in $E(S)$ since $x\in F$. We know that $wy\in E(S)$ and there is a triangle $wyz\in S$. For this case, Reduction Rule \ref{increasing3} could be applied to replace $vuw$ and $ywz$ with  $xvu$ and $wyz$ (Case 3 in Figure \ref{f3}).

In any of these cases, we can find a contradiction to the fact that the graph is reduced.
 \hfill \qed\par
\end{proof}

 \begin{figure}[t]
     \centering
     \includegraphics[scale=0.3]{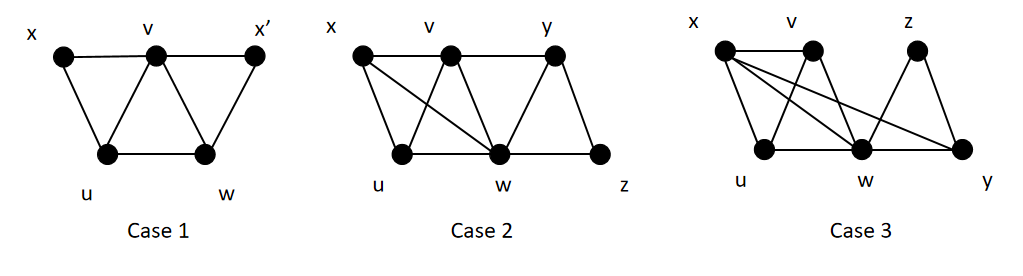}
     \vspace{-3mm}\caption{Three cases in Lemma \ref{Loe}}

     \label{f3}
 \end{figure}

A triangle $uvw\in S$ is \emph{good} if it contains a labeled edge and \emph{bad} otherwise.
By Lemma \ref{Loe}, we know that there is exactly one labeled edge in each good triangle.
We let $G'$ be the graph obtained by deleting the set $L$ of labeled edges from $G$. Consider a good triangle $uvw$ with labeled edge $uv$. If the two edges $vw$ and $wu$ are not in any triangle in $G'$, we call the triangle \emph{excellent}. Otherwise, we call the triangle \emph{pretty-good}.
We let $S_1$ denote the set of excellent triangles, $S_2$ denote the set of pretty-good triangles, and
$S_3$ denotes the set of bad triangles in $S$.
The number of triangles in $S_1$, $S_2$ and $S_3$ are denoted by $k_1,k_2$, and $k_3$, respectively.
Let $V_1=V(S_1)\setminus (V(L)\cup V(S_2)\cup V(S_3))$ and $V_2=V(S_2)\setminus V(L)$.
 See Figure \ref{l2} for an illustration of these concepts.
\begin{figure}[t]
    \centering
    \includegraphics[scale=0.30]{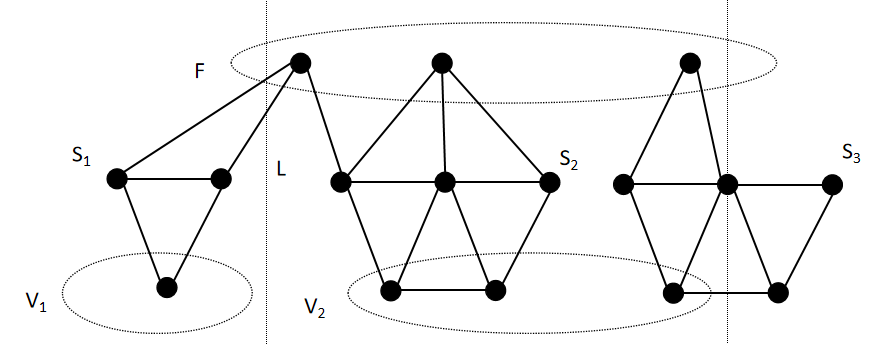}
    \caption{An illustration for triangles and vertices in $G$}
    \label{l2}
\end{figure}

\vspace{-4mm}
\subsection{The analysis}
The discharging method stands as a renowned technique in graph theory, finding its most notable application in the proof of the famous Four Color Theorem.
In this section, we will use the discharging method to analyze the number of vertices present in $S_1$, $S_2$, $S_3$, and $F$.
The idea of the method is as follows.

First, we initially assign some integer values to vertices, edges, and triangles in $S$.
The total value assigned is at most $3k$. Subsequently, we perform steps to update the values, where certain values on vertices, edges, and triangles are transformed into other vertices, edges, and triangles.
In these steps, we never change the structure of the graph and the total value in the graph. After performing these transformations, we demonstrate that each vertex in the graph has a value of at least 1.
Consequently, we conclude that the number of vertices in the graph is at most $3k$. 

\textbf{Initialization:} Assign value 3 to each edge in $L$ and each triangle in $S_3$. Edges not in $L$, vertices, and triangles in $S_1\cup S_2$ are assigned a value of 0.

By Lemma~\ref{Loe}, we know that each of excellent and pretty-good triangles contains exactly one labeled edge in $L$ and each bad triangle in $S_3$ contains no labeled edge. Thus, the total value in the graph
is $3k_1+3k_2+3k_3\leq 3k$.

\textbf{Step 1:} For each labeled edge in $L$, transform a value of 1 to each of its two endpoints; for each triangle in $S_3$, transform a value of 1 to each of its three vertices.

Figure~\ref{fs1} illustrates the transformation process of Step 1. After Step 1, each labeled edge has a value of 1, and all triangles have values of 0.
Note that some vertices may have a value of
 2 or more, as they may serve as endpoints of multiple labeled edges and can also be vertices in $V(S_3)$.
However, vertices in $F\cup V_1\cup V_2$ still retain a value of 0.

\begin{figure}[t]
    \centering
    \includegraphics[scale=0.25]{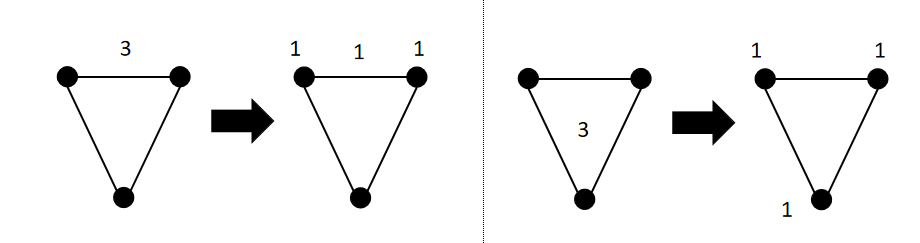}
    \caption{An illustration for Step 1}
    \label{fs1}
\end{figure}

A \emph{triangle component} is a connected component in the graph $H=(V(S),E(S))$. For a vertex $v\in V(S)$, we let $C(v)$ denote the set of vertices in the triangle component which contain $v$.

\textbf{Step 2:} For each triangle component in $G$, we iteratively transform a value of 1 from a vertex with a value of at least 2 to a vertex with a value of 0 in the same triangle component, where vertices in $V_1$ have a higher priority to get the value.
\begin{lemma}\label{a2}
After Step 2, each triangle component has at most one vertex with a value of 0. Moreover,

(i) For any triangle component containing a triangle in $S_3$, each vertex in the triangle component has a value of at least 1;

(ii) For any triangle component containing at least one triangle in $S_2$, if there is a vertex with a value of 0 in the triangle component, then the vertex must be a vertex in $V_2$.

\end{lemma}
\begin{proof}
Let $Q$ be a triangle component with $x$ triangles. Since $Q$ is connected, it contains at most $2x+1$ vertices.
Assume that among the $x$ triangles, there are $x_1$ triangles in $S_1\cup S_2$ and $x_2$ triangles in $S_3$, where $x_1+x_2=x$.
By the definition, we know that each triangle in $S_1\cup S_2$ contains a distinct labeled edge. According to the initialization of the assignment,
we know that the total value is $2x_1+3x_2=2x+x_2$. It always holds that $2x+1\leq (2x+x_2)+1$, and $2x+1\leq 2x+x_2$ when $x_2\geq 1$.
Thus, $Q$ has at most one vertex with a value of 0. When $Q$ contains some triangles from $S_3$, i.e., $x_2\geq 1$, all vertices in $Q$ will get a value of at least 1.
The statement (ii) holds because vertices in $V_1$ have a higher priority to receive the value in Step 2.
\hfill \qed\par
\end{proof}

After Step 2, only vertices in $F$, some vertices in $V_1$, and some vertices in $V_2$ have values of 0. We use the following lemma to transform some values to vertices in $V_2$ with a value of 0.

\begin{lemma}\label{ano}
 Consider two triangles $vuw$ and $vxy\in S_2$ sharing a common vertex $v$, where $uv$ and $vx\in L$. If there is an edge $wy\in E$, then $uv$ and $vx$ are spanned by exactly one vertex in $F\cup V_1\setminus C(v)$.
\end{lemma}
\begin{figure}[h]
    \centering
    \includegraphics[scale=0.3]{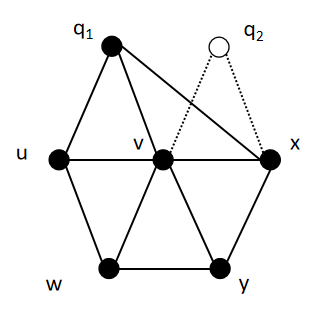}
    \caption{An illustration for Lemma \ref{ano}}
    \label{fs2}
\end{figure}

\begin{proof}
Since $vu$ and $vx$ are labeled edges in $L$, we know that each of $vu$ and $vx$ is spanned by at least one vertex in $F$.
Assume to the contrary that there are two different vertices $q_1$ and $q_2\in F\cup V_1\setminus C(v)$ that span edges $uv$ and $vx$, respectively. We show some contradiction. First, none edge in $\{q_1u,q_1v,q_2v,q_2x\}$ belongs to $E(S)$ since $q_1$ and $q_2$ are not part of $C(v)$.
Thus, Reduction Rule \ref{increasing2} could be applied to replace triangles $uvw$ and $vxy$ with triangles $q_1uv, q_2vx$, and $vwy$, a contradiction to the factor that the graph is reduced. See Figure~\ref{fs2} for an illustration of the proof.
\hfill \qed\par
\end{proof}

\textbf{Step 3:}
If there are two triangles $vuw$ and $vxy\in S_2$ sharing a common vertex $v$ such that $uv$ and $vx$ are two labeled edges in $L$ and there is an edge $wy\in E$, we transform a value of 1 from edge $vx$ to the unique vertex $q_1\in F$ spanning $vx$ and transform a value of 1 from edge $uv$ to the vertex with a value of 0 in $C(v)$ if this vertex exists.
\begin{figure}[t]
    \centering
    \includegraphics[scale=0.30]{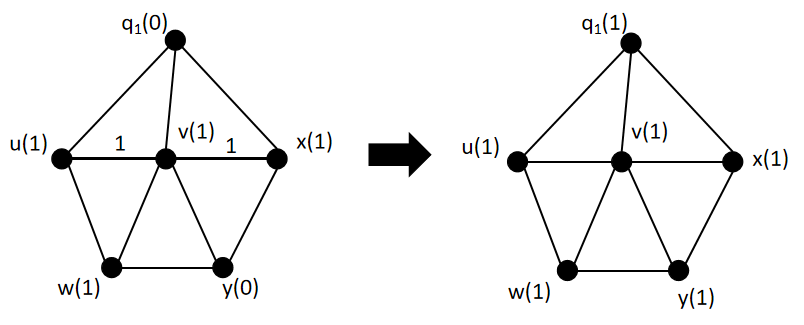}
    \caption{An illustration for Step 3, where the number in parentheses next to each vertex represents the value of that vertex}
    \label{fs3}
\end{figure}

See Figure~\ref{fs3} for an illustration of Step 3. We have the following property.
\begin{figure}[t]
    \centering
    \includegraphics[scale=0.18]{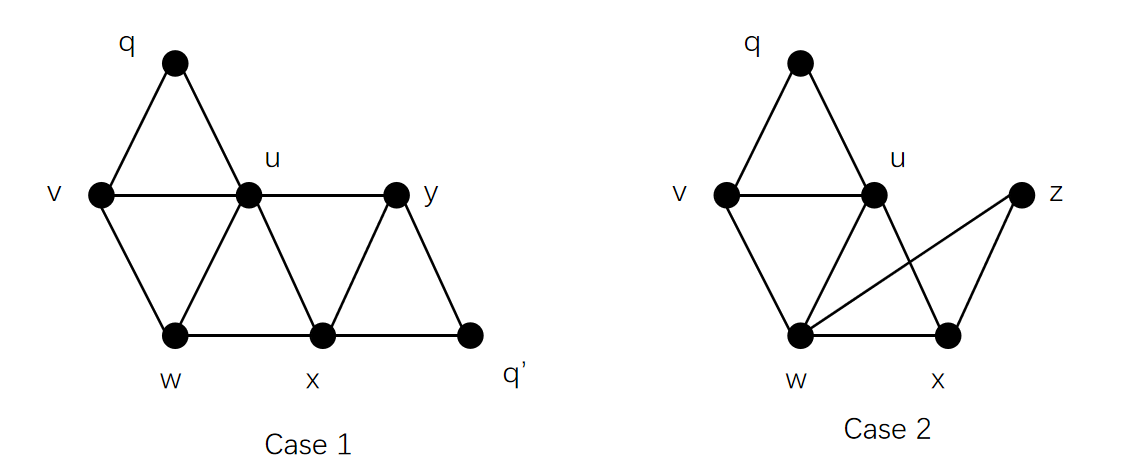}
    \caption{An illustration for Lemma \ref{nov2}}
    \label{fl8}
\end{figure}
\begin{lemma}\label{nov2}
Every vertex in $V_2$ has a value of at least 1 after Step 3.
\end{lemma}

\begin{proof}
Assume to the contrary there is a vertex $w\in V_2$ with a value of 0. We know that all vertices in $C(w)\setminus \{w\}$ have a value of at least 1 by Lemma \ref{a2}.
Let $wuv\in S_2$ be the triangle containing $w$, where $uv$ is the labeled edge spanning by a vertex $q\in F$. As shown in Figure \ref{fl8}. At least one of $uw$ and $vw$ is in a triangle in graph $G-L$ by the definition of $S_2$. Without loss of generality, we assume that $uw$ is contained in a triangle $uwx$, where $x\neq v$.
At least one of $ux$ and $wx$ is contained in a triangle in $S$ otherwise Reduction Rule \ref{increasing} could be applied.

\textbf{Case 1}: Edge $ux$ is contained in a triangle $uxy\in S$. See Case 1 in Figure \ref{fl8}. The triangle $uxy$ is not in $S_3$ otherwise there is a contradiction that $w$ would have a value of at least 1 by Lemma~\ref{a2}. Thus, triangle $uxy$ must be in $S_1\cup S_2$ and there is exactly one of $ux$, $uy$, and $xy$ is a labeled edge. Edge $ux$ would not be a labeled edge since triangle $uwx$ is contained in $G-L$.
If $xy$ is the labeled edge that is spanned by a vertex $q'\in F$, then
Reduction Rule \ref{increasing3} could be applied to replace triangles $uvw$ and $uxy$ with triangles $uvw$ and $xyq'$, a contradiction to the factor that the graph is reduced.
If $uy$ is the labeled edge, then triangle $uxy\in S_2$ since $ux$ is contained in a triangle $uxw\in G-L$.
For this case, the vertex $w$ would have a value of at least 1 by Lemma~\ref{a2}. We can always find a contradiction.

\textbf{Case 2}: Edge $wx$ is contained in a triangle $wxz\in S$. See Case 2 in Figure \ref{fl8}. If $z\neq q$, then Reduction Rule \ref{increasing3} could be applied to replace triangles $uvw$ and $wxz$ with triangles $wxz$ and $qvu$, a contradiction to the factor that the graph is reduced. If $z=q$, then at least two edges in triangle $vuw$ are spanned by vertices in $F$, a contradiction to Lemma~\ref{Loe}.

In either case, a contradiction is reached, which implies that the assumption of a vertex $w \in V_2$ having a value of 0 is incorrect. Therefore, every vertex in $V_2$ has a value of at least 1 after Step 3.
\hfill \qed\par
\end{proof}
After Step 3, all vertices with a value of 0 are in either $F$ or $V_1$.
We let $V_1'$ denote the set of vertices with a value of 0 in $V_1$, $F'$ denote the set of vertices with a value of 0 in $F$, and $L'$ denote the set of edges with a value of 1 in $L$ after Step 3.
We give more properties.

\begin{lemma}\label{independent}
Set $F'\cup V'_1$ is an independent set.
\end{lemma}

\begin{proof}
We prove that $F\cup V_1$ is an independent set, which implies $F'\cup V_1'$ is an independent set since $F'\cup V_1'\subseteq F\cup V_1$.
Assume to the contrary that there is an edge $uv$ between two vertices in $F\cup V_1$. There is at least one triangle $uvw\in G$ containing $uv$ since Reduction Rule \ref{rule1} has been applied.
At least one of $uv, vw$, and $uw$ must be in $E(S)$ by the maximality of $S$.

If $uv\in E(S)$, we let $uvx$ be the triangle in $S$ containing $uv$. First, we know that $u$ and $v\in V_1$ since $F\cap V(S)=\emptyset$.
By the definition of $V_1$,
we get that none of $u$ and $v$ is contained in a triangle in $S_3$ and none of $u$ and $v$ is an endpoint of a labeled edge. Thus, triangle $uvx$ is not a triangle in $S_3$ and it does not contain any labeled edge and then it is not a triangle in $S_1\cup S_2$, which implies triangle $uvx$ is not in $S$, a contradiction.

Otherwise, one of $uw$ and $vw$, say $uw$, is contained in $E(S)$.
Let $uwy$ be the triangle in $S$ containing $uw$. We also have that $u\in V_1$ since $F\cap V(S)=\emptyset$.
By the definition of $V_1$, we know that $u$ is not a vertex in a triangle in $S_3$, and then $uwy$ is not a triangle in $S_3$.
Thus, triangle $uwy$ can only be in $S_1\cup S_2$. Note that none of $uv$, $vw$ and $uw$ can be a labeled edge since $u$ and $v\in F\cup V_1$.
Thus, edge $uw$ is still in a triangle $uvw$ in $G-L$, and then $uwy$ can not be a triangle in $S_1$. However, triangle $uvw$ can not be a triangle in $S_2$ too since $u$ is a vertex in $V_1$.
We also get a contradiction that triangle $uvw$ is not in $S$.

Hence, we have shown that no edge exists between any two vertices in $F' \cup V'_1$, which proves that $F' \cup V'_1$ forms an independent set.
\hfill \qed\par\end{proof}

\begin{lemma}\label{fl}
 Vertices in $F'$ only span edges in $L'$.
\end{lemma}
\begin{proof}
Assume to the contrary that there is a vertex $z\in F'$ that spans an edge $e\notin L'$.
Note that $z\in F$  can only span edges in $L$. Thus, $e\in L\setminus L'$, i.e., the value of $e$ is 0 now.
Only Step 3 will create labeled edges with a value of 0. Therefore, edge $e$ appears as a labeled edge in the graph structure described in Step 3. By Lemma \ref{ano}, we know that $z$ is the unique vertex in $F$ spanning $e$ and the value of $z$ is at least 1 after Step 3, a contradiction to $z\in F'$.
Thus, vertices $F'$ can only span edges in $L'$.
\hfill \qed\par \end{proof}

\begin{lemma}\label{v1l}
 Vertices in $V_1'$ only span edges in $L'$.
\end{lemma}
\begin{figure}[t]
    \centering
    \includegraphics[scale=0.25]{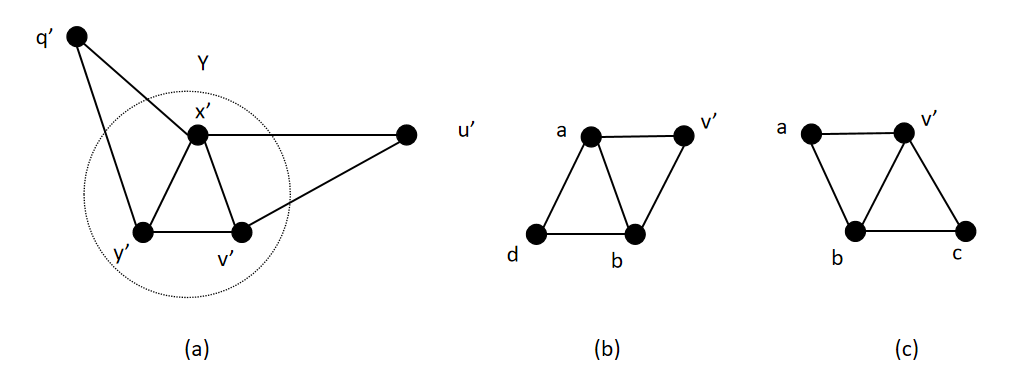}
    \caption{Some cases in the proof of Lemma \ref{v1l}}
    \label{fl11}
\end{figure}
\begin{proof}
Let $v'\in V_1'$ and $Y$ be the triangle component containing $v'$. First, we show that $N(v')\subseteq V(Y)$.

We let $E_1$ denote the set of edges with one endpoint being $v'$ and the other endpoint being in $N(v')\cap V(Y)$ and $E_2$ denote the set of edges with one endpoint being $v'$ and the other endpoint being in $N(v')\setminus V(Y)$.
Assume that there are edges $x'v'\in E_1$ and $u'v'\in E_2$ forming a triangle $x'v'u'$ in $G$ and we show some contradiction. If $x'v'\notin E(S)$, then none of $x'v'$, $u'v'$, and $x'u'$ is contained in $E(S)$
since $u'$ is not in the triangle component containing $v'$. which is a contradiction to the maximality of $S$. Otherwise $x'v'\in E(S)$ and we let triangle $x'v'y'\in S$ be the triangle containing $x'v'$.
Since the value of $v_1$ is 0, by Lemma~\ref{a2}, we know that $Y$ could not contain any triangle in $S_2\cup S_3$ and then $x'v'y'\in S_1$.
We know that $x'y'$ is the labeled edge since $v'\in V_1'$. Let $q'\in F$ be the vertex spanning $x'y'$.
We can see that $q'\neq u'$ otherwise all the three edges in triangle $x'y'v'$ are labeled edges spanned by a vertex $q'\in F$, a contradiction to Lemma~\ref{Loe}.
However, for this case, Reduction Rule \ref{increasing} could be applied to replace triangle $x'v'y'$ with triangles $q'x'y'$ and $u'x'v'$ in $S$, a contradiction to the factor that the graph is reduced.
Therefore, no triangle in $G$ contains edges both from $E_1$ and $E_2$. If both of $E_1$ and $E_2$ are not empty, then the condition of Reduction Rule~\ref{split} would hold, a contradiction.
Since $E_1$ can not be empty, we know that $E_2$ is empty. Thus, we have  $N(v')\subseteq V(Y)$. See Figure~\ref{fl11}(a) for an illustration of the proof.

Assuming that $v'$ spans an edge $ab$, we can conclude that both $a$ and $b$ are in $V(Y)$ since $N(v')$ is a subset of $V(Y)$. Our goal is to prove that $ab$ belongs to $L$.
If $ab\in E(S)\setminus L$, then $ab$ is contained in a triangle $abd\in S_1$ since $Y$ contains only triangles in $S_1$.
We have $d\neq v'$, otherwise the triangle $abv'\in S_1$ would have no labeled edge, a contradiction.
The triangle $abd\in S_1$ contains an edge $ab$ which will be in a triangle $abv'$ in $G-L$, which contradicts the definition of $S_1$ (See Figure~\ref{fl11}(b)).
If $ab\not\in E(S)$, then at least one of $bv'$ and $av'$, say $bv'$, is contained in $E(S)$ by the maximality of $S$.
 Let $bv'c\in S_1$ be the triangle containing $bv'$ where $c\neq a$ since $ab\not\in E(S)$.
  The triangle $bcv'\in S_1$ contains an edge $bv'$ which will be in a triangle $abv'$ in $G-L$, which contradicts the definition of $S_1$ (See Figure~\ref{fl11}(c)).
  Thus, it holds that $ab\in L$.
Any vertex $v'\in V_1'$ only span edges $ab \in L$, where both $a$ and $b$ are in $V(Y)$. Then, edge $ab$ can only be a labeled edge in a triangle in $S_1$.
Since all labeled edges in triangles in $S_1$ are in $L'$, we know that vertices in $V_1'$ can only span edges in $L'$.
\hfill \qed\par\end{proof}

\begin{lemma}\label{fv1}
After Step~3, it holds that $|F'\cup V_1'|\leq |L'|$.
\end{lemma}
\begin{proof}
By Lemma~\ref{independent},~\ref{fl}, and~\ref{v1l}, we know that $F'\cup V_1'$ is an independent set and any vertex $v'\in F'\cup V_1'$ only span edges in $L'$.
If $|F'\cup V_1'|> |L'|$, then by Lemma~\ref{lemma_cr} there is a fat-head crown decomposition $(C,H,X)$ of $G$ with $C\subseteq F'\cup V_1'$ and $H\subseteq L'$. Moreover, the fat-head crown decomposition can be detected by Lemma \ref{FCD_exist} and will be handled by Reduction Rule \ref{crown} since $F'\cup V_1'\subseteq F\cup V_1\subseteq V\setminus V(L)$ and $L'\subseteq L$. Thus, we know the lemma holds.
\hfill \qed\par
\end{proof}

\textbf{Step 4:} We transform value from edges in $L'$ to vertices in $F'\cup V_1'$ such that each vertex in $F'\cup V_1'$ gets value at least 1 by Lemma \ref{fv1}.

After Step 4, each vertex in $G$ has a value of at least 1. Since the total value in $G$ is at most $3k$, we can conclude that the graph has at most $3k$ vertices.
\begin{theorem}
\textsc{Edge Triangle Packing} and  \textsc{Edge Triangle Covering} admit a kernel of at most $3k$ vertices.
\end{theorem}
\section{Conclusion}\label{sec_con}
In this paper, we present simultaneous improvements in the kernel results for both \textsc{Edge Triangle Packing} and \textsc{Edge Triangle Covering}. Our approach incorporates two key techniques to achieve these enhancements.
The first technique involves utilizing fat-head crown decomposition, which enables us to effectively reduce various graph structures. By applying this technique, we can simplify the problem instances.
The second technique we introduce is the discharging method, which plays a crucial role in analyzing kernel size. This method is simple and intuitive, and we believe it has the potential to be applied to the analysis of other kernel algorithms.
\newpage

\noindent\textbf{Acknowledgments.} \rm{The work is supported by the National Natural Science Foundation of China, under the grants 62372095 and 61972070.}

\bibliography{ETP}

\end{document}